\documentclass[11pt]{article}
  
\usepackage{cite}
\usepackage{graphicx}
\usepackage{amsmath}
\usepackage{amsthm}
\usepackage{yhmath} 

\usepackage{amsfonts} 
\usepackage{amssymb}
\usepackage{fullpage}
\usepackage{color}
\usepackage{latexsym}

\usepackage{algorithm}
\usepackage[noend]{algorithmic}
\usepackage{multirow} 

\newcommand{\old}[1]{{}}

\newtheorem{theorem}{Theorem}[section]
\newtheorem{corollary}[theorem]{Corollary}
\newtheorem{lemma}[theorem]{Lemma}
\newtheorem{claim}[theorem]{Claim}

\newcommand{\mono}{{{monochromatic }}}

\newcommand{\T}{{{\cal{T}}}}

\newcommand{\N}{{{\cal{N}}}}

\title{Planar Bichromatic Bottleneck Spanning Trees \footnote{This work was partially supported by Grant 2016116 from the United States -- Israel Binational Science Foundation. Mitchell was partially supported by NSF (CCF-1526406).}
} 
\author{ A. Karim Abu-Affash\thanks{Software Engineering Department, Shamoon College of Engineering, Beer-Sheva 84100, Israel, {\tt abuaa1@sce.ac.il}.}
\and 
Sujoy Bhore\thanks{Department of Computer Science, Ben-Gurion University, Beer-Sheva 84105, Israel, {\tt sujoy.bhore@gmail.com}. }
\and 
Paz Carmi\thanks{Department of Computer Science, Ben-Gurion University, Beer-Sheva 84105, Israel, {\tt carmip@cs.bgu.ac.il}. }
\and
Joseph S. B. Mitchell\thanks{Department of Applied Mathematics and Statistics, Stony Brook University, Stony Brook, NY, USA, {\tt jsbm@ams.stonybrook.edu}.}
}

\begin{document}

\maketitle

\begin{abstract}

Given a set $P$ of $n$ red and blue points in the plane, a \emph{planar bichromatic spanning tree} of $P$ is a spanning tree of $P$, such that each edge connects between a red and a blue point, and no two edges intersect. 
In the bottleneck planar bichromatic spanning tree problem, the goal is to find a planar bichromatic spanning tree $T$, such that the length of the longest edge in $T$ is minimized. 
In this paper, we show that this problem is NP-hard for points in general position. Moreover, we present a polynomial-time $(8\sqrt{2})$-approximation algorithm, by showing that any bichromatic spanning tree of bottleneck $\lambda$ can be converted to a planar bichromatic spanning tree of bottleneck at most $8\sqrt{2}\lambda$.
\end{abstract}


\section{Introduction}

Let $P$ be a bi-colored set of red and blue points in the plane and let $n=|P|$. A \emph{bichromatic spanning tree} of $P$ is a spanning tree of $P$ whose edges are straight-line segments connecting between points of different colors. A spanning tree is {\em planar} if its edges do not cross each other.
Borgelt et al.~\cite{Borgelt09} studied the problem of computing a minimum-weight planar bichromatic spanning tree, and showed that the problem is NP-hard. Moreover, for points in general position, they gave an $O(\sqrt{n})$-approximation algorithm, and for points in convex position, they gave an exact cubic-time algorithm. 
Biniaz et al.~\cite{Biniaz19} studied the problem of computing a maximum-weight planar bichromatic spanning tree and gave a $(1/4)$-approximation algorithm for the problem.
 
Algorithmic problems on bichromatic geometric input have appeared in many problems, including, e.g.,
trees~\cite{Abellanas99,Biniaz18,Borgelt09}, matchings~\cite{BiniazMS14,Dumitrescu01}, and partitionings~\cite{Dumitrescu02}. 
Often the bichromatic input is referred to as ``red-blue'' input, e.g. in red-blue intersection~\cite{Agarwal90,Mairson88}, red-blue separation~\cite{Arora04,BoissonnatCDUY00,DemaineEHILMOW05,EverettRK96}, and red-blue connection problems~\cite{Agarwal91,Atallah01}. For a survey of many geometric problems on bichromatic (red-blue) points, see Kaneko and Kano~\cite{Kaneko03}. 

Many of the structures studied in computational geometry are planar,
including minimum spanning trees, minimum weight matchings, Delaunay/Voronoi diagrams, etc.
Therefore, the planarity requirement is quite natural, and indeed many researchers have considered geometric problems dealing with crossing-free configurations in the plane; see, e.g.~\cite{Abu-Affash15,Abu-Affash14,Aichholzer08,Alon93,Aloupis10}.

In this paper, we study the problem of computing a bottleneck planar bichromatic spanning tree of $P$, in which we seek a planar bichromatic spanning tree that minimizes bottleneck, i.e., the length of the longest edge. 
To the best of our knowledge, this problem has not been considered before.

\vspace{0.2cm}
\noindent
\textbf{Our results.}
In Section~\ref{sec:hardness}, we prove that the problem of computing a bottleneck planar bichromatic spanning tree is NP-hard. 
Our proof is based on a reduction from the planar 3-SAT problem, and is influenced by the proof of Borgelt et al.~\cite{Borgelt09}.
As a corollary we obtain that the problem does not admit a PTAS.
Next, in Section~\ref{sec:approx}, we present a polynomial-time algorithm for computing a planar bichromatic spanning tree of bottleneck at most $8\sqrt{2}$ times the bottleneck of a bottleneck planar bichromatic spanning tree. 
We first compute a bottleneck bichromatic spanning tree having bottleneck $\lambda$ that may have crossings. Then, we use the length $\lambda$ to define a partition of the plane into convex cells, and to partition $P$ into subsets according to these cells. Next, we construct planar bichromatic trees for each subset, and we connect these trees to obtain a planar bichromatic spanning tree of $P$. We show that this tree has a bottleneck at most $8\sqrt{2}\lambda$.

\section{Hardness Proof} \label{sec:hardness}
In this section, we prove the following theorem.
\begin{theorem} \label{thm:hardness}
Let $P$ be a set of $n$ red and blue points in the plane. Then, computing a bottleneck planar bichromatic spanning tree of $P$ is NP-hard.
\end{theorem}
\begin{proof}
We adapt the proof of Borgelt et al.~\cite{Borgelt09}, making modifications necessary to address the bottleneck version.
For completeness, we explain the ingredients required for the proof.

The proof is based on a reduction from the planar 3-SAT problem. 
Given a 3-CNF formula $F$ with $n$ variables $X=\{x_1,x_2,\ldots,x_n\}$ and $m$ clauses $Y=\{C_1,C_2,\ldots,C_m\}$,
let $G_F=(V,E)$ be the graph of $F$, i.e., $V=X \cup Y$ and $E=\{(x_i,C_j):x_i \mbox{ appears in } C_j\}$.
If $G_F$ is planar,
then $F$ is called a planar 3-CNF formula. The planar 3-SAT problem is to determine whether a given planar 3-CNF formula $F$ is satisfiable; the problem is NP-complete~\cite{Lichtenstein82}.

Let $F$ be a planar 3-SAT formula. We construct, in polynomial time, a set $P$ of red and blue points in the plane, such that $F$ is satisfiable if and only if there exists a planar bichromatic spanning tree of $P$ of bottleneck 1.
Consider the graph $G_F$. It is well known that $G_F$ can be embedded in the plane in polynomial time.

The construction is based on chains of pairs of red and blue points. We call the pairs in the chain \emph{sites}. The distance between the points in each site is less than 1, and the distance between two points of different colors in consecutive sites is exactly 1; see Figure~\ref{fig:chain}(a). Now, for every two consecutive sites, there are two possible edges to connect them: we either connect the blue point of the first site with the red point of the second site, or the other way around. 
Moreover, the chain is constructed in such a way that if we connect the blue point in the leftmost site to the red point in the next site, this forces the choice of connections in one direction along the chain; see Figure~\ref{fig:chain}(b).

\begin{figure}[ht]
   \centering
       \includegraphics[width=0.86\textwidth]{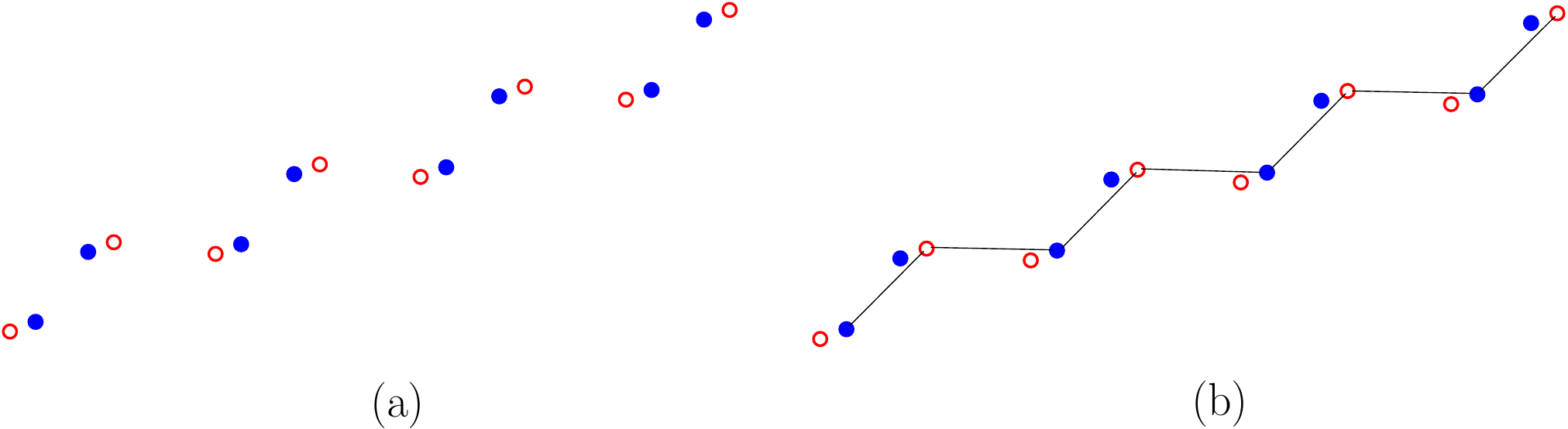}
   \caption{A chain of red-blue sites. }              
   \label{fig:chain}
\end{figure}

\textbf{Variables.} Each variable $x_i \in X$ is represented by a circular chain of sites, 
a special red point $r_i$, and a red-blue path; see Figure~\ref{fig:variables}. The addition of the red-blue path to the variable gadget of~\cite{Borgelt09} is required, since without it the special red point $r_i$ can be connected to both sides of the chain without increasing the bottleneck, which is not the case in the minimum weight version. The red-blue path forces $r_i$ to be connected exactly to one side.
\begin{figure}[ht]
   \centering
       \includegraphics[width=0.76\textwidth]{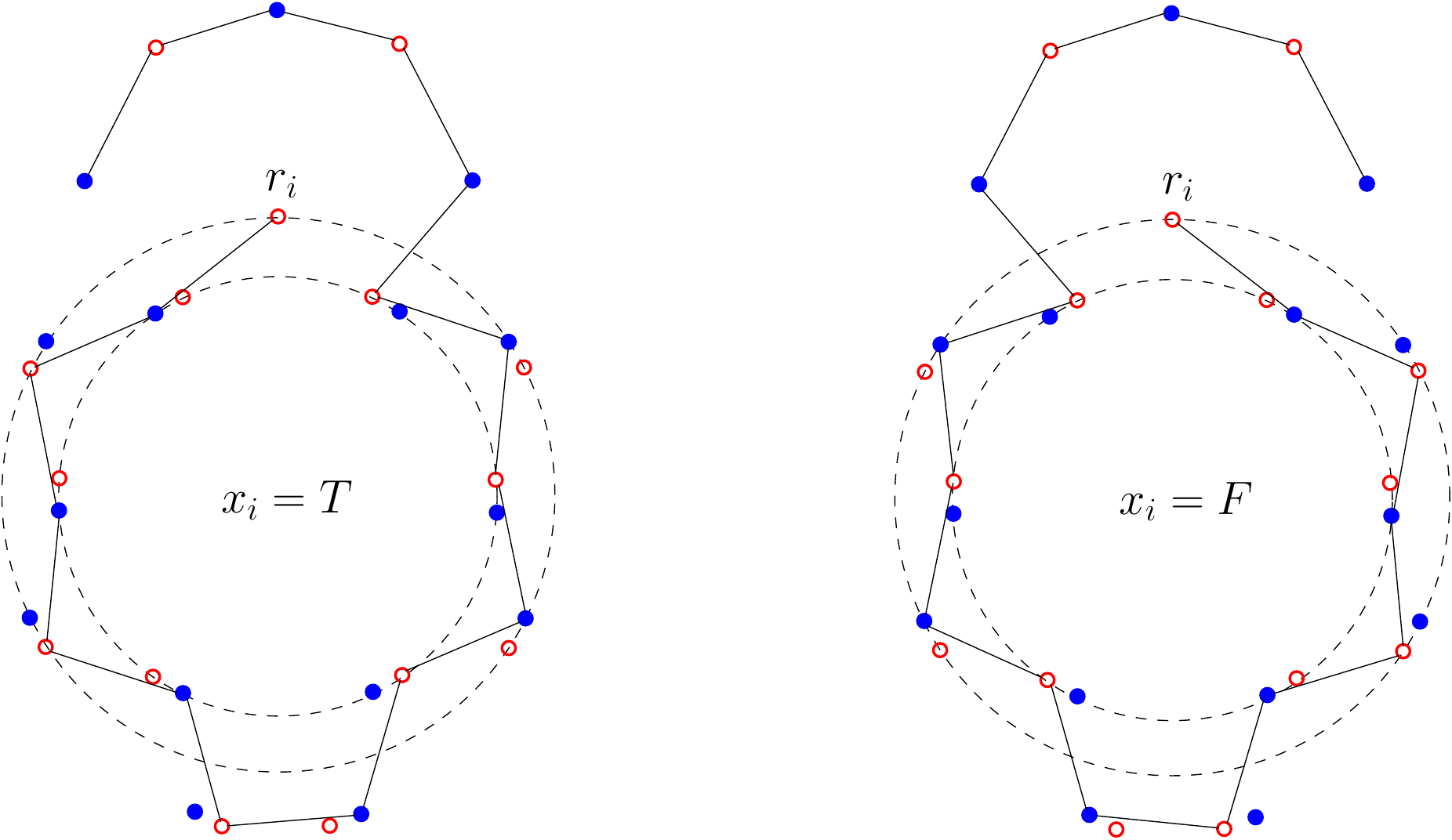}
   \caption{The trees corresponding to the true and the false assignments of $x_i$. }              
   \label{fig:variables}
\end{figure}

The sites on the circular chain are located on two (inner and outer) circles. We locate the points in such a way that the distance between each consecutive sites on the circular chain, the distance between $r_i$ and its neighboring sites in the chain, and the distance between the endpoints of the red-blue path and the chain is exactly 1. Moreover, the points are located such that, there are only two possible optimal trees (i.e., planar bichromatic spanning trees of bottleneck 1) of the points, depending on the connection of $r_i$ to the chain. In both trees, $r_i$ is connected to exactly one site of the chain. We arbitrarily associate one of them with the assignment $x_i=T$, and the other with the assignment $x_i=F$; see Figure~\ref{fig:variables}. Thus, the value of $x_i$ will determine the tree of these points, and vice virsa. Moreover, if $x_i=T$ (resp., $x_i=F$), then the red points on the right (resp., on the left) of the inner circle are free to be connected to points outside the gadget, and the red points on the left (resp., on the right) of the inner circle cannot be connected to points outside the gadget without crossing, and vice versa.

\textbf{Clauses.} Each clause $C_j$ is represented by a single red point $r_j$ and three chains that will be connected to the respective variables of $C_j$; see Figure~\ref{fig:clauses}(a). The distance between $r_j$ and each chain is 1. In any optimal tree, $r_j$ will be connected to at least one of the three chains. However, it cannot be connected to any chain if all the chains are connected to variables that are in the wrong state; see Figure~\ref{fig:clauses}(b).

\begin{figure}[ht]
   \centering
       \includegraphics[width=0.89\textwidth]{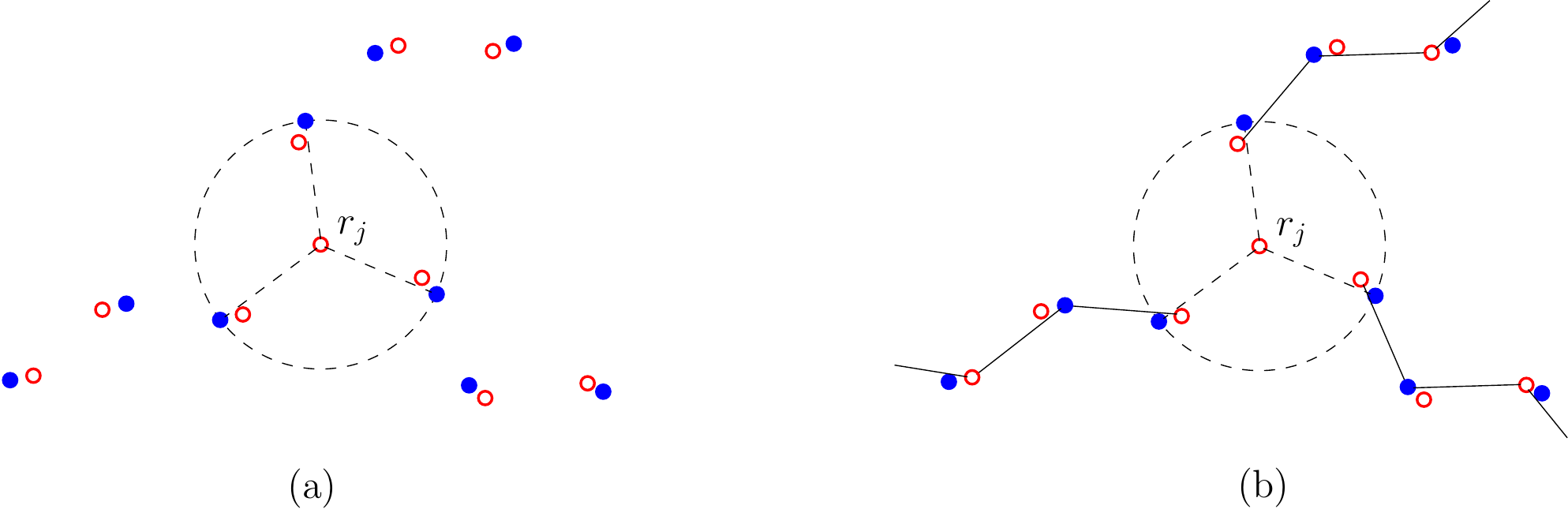}
   \caption{The gadget corresponding to the clause $C_j$. }              
   \label{fig:clauses}
\end{figure}

We connect between the variables and the clauses such that, in any optimal tree, one of the three chains of the clause has to be connected to a red point on the inner circle of the corresponding variable. Assume that $x_i$ appears unnegated in a clause $C_j$ and negated in a clause $C_k$, i.e., $C_j=(x_i\vee\cdot\vee\cdot)$ and $C_k=(\widetilde{x_i}\vee\cdot\vee\cdot)$. We connect the chain of $C_j$ that is respective to $x_i$ to a site on the right of the inner circle of the gadget $x_i$, and we connect the chain of $C_k$ that is respective to $x_i$ to a site on the left of the inner circle of the gadget $x_i$; see Figure~\ref{fig:gadgetN}. This connection ensures that, if $x_i$ is assigned $T$, then the red point on the right of the inner circle of $x_i$ is free to be connected to the chain of $C_j$, and this connection can produce a path through the chain that ends at $r_j$. On the other hand, if $x_i$ is assigned $T$, then the red point on the left of the inner circle of $x_i$ cannot be connected to the chain of $C_k$, which does not allow a connection between the chain and $r_k$. The same argument holds when $x_i$ is assigned $F$.

\begin{figure}[ht]
   \centering
       \includegraphics[width=0.83\textwidth]{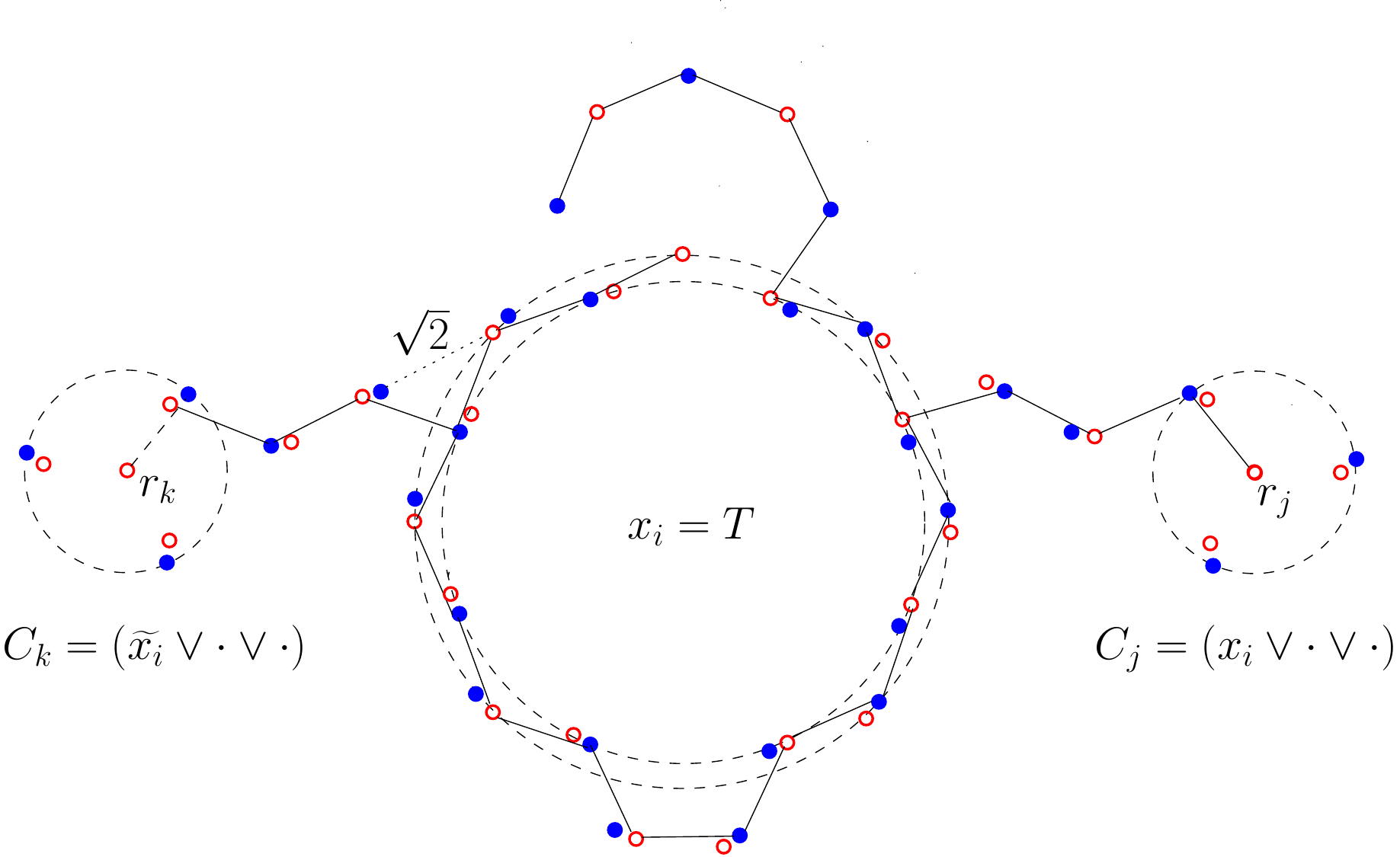}
   \caption{The connection between the variable $x_i$ and the clauses $C_j$ and $C_k$. }              
   \label{fig:gadgetN}
\end{figure}

Finally we need to connect all variables to each other by some fixed part of the tree, because the whole construction needs to be a tree and not a forest. These connections can easily be made using red-blue chains having distance at most 1 between every two consecutive points in the path. Also, we need to make sure that the distance between different parts of the construction is large enough to avoid shortcuts.
\end{proof}

Notice that in the reduction we proved that if the 3-SAT formula is not satisfiable, then any planar bichromatic spanning tree of $P$ has an edge of length greater than 1. Actually, we can push the length of this edge to be closed to $\sqrt{2}$. 
That is, we can draw the connection between each clause and its corresponding variables, such that 
the distance between each chain of the clause and the corresponding site on the inner circle of the variable is 1, and the distance between each chain of the clause and the sites on the outer circle of the variable is at least $\sqrt{2} - \varepsilon $, for any $0 < \varepsilon < \sqrt{2}-1$; see Figure~\ref{fig:gadgetN}.
This implies that the bottleneck planar bichromatic spanning tree problem cannot be approximated within a factor less than $\sqrt{2}$, unless $P=NP$.

\begin{corollary}
The bottleneck planar bichromatic spanning tree problem cannot be approximated within a factor less than $\sqrt{2}$, unless $P=NP$. In particular, there is no PTAS (unless $P=NP$).
\end{corollary}


\section{Approximation Algorithm} \label{sec:approx}
Let $P$ be a set of red and blue points in the plane and let $n=|P|$. 
Let $T$ be a bichromatic spanning tree of $P$ of minimum bottleneck ($T$ may have crossings and can be computed in $O(n\log{n})$ time~\cite{Biniaz18}).
Let $\lambda$ denote the bottleneck of $T$, i.e., the length of the longest edge in $T$. 
Notice that $\lambda$ is the lower bound for any bichromatic spanning tree of $P$, in particular for any planar bottleneck bichromatic spanning tree of $P$.
In this section, we show how to compute a planar bichromatic spanning tree of $P$, such that its bottleneck is at most $8\sqrt{2} \lambda$.

Our algorithm partitions the plane into disjoint cells satisfying the following properties:
\begin{enumerate}
	\item Each cell is convex and contains points of both colors.
	\item In each cell, the distance between any two points is bounded by $5\sqrt{2}\lambda$.
	\item The cells are connected, i.e., if we consider the graph with the cells as its vertices and there is an edge between two cells if they are adjacent (sharing a common boundary), then this graph is connected.
	\item We can construct a planar bichromatic spanning tree of the points in each cell and we can connect them without crossings.
\end{enumerate}

Assume, w.l.o.g., that $\lambda=1$. We begin by laying an axis-parallel grid,
such that each (square) cell is of edge length $3$ and the points of $P$ are in the interior (not on the boundary) of these cells; see Figure~\ref{fig:grid1}. 
We say that a cell $C_{i,j}$ is bichromatic if it contains points of both colors and we say that $C_{i,j}$ is monochromatic (red or blue) if all of the points in $C_{i,j}$ have the same color, otherwise, we say that $C_{i,j}$ is an empty cell.
\begin{figure}[ht]
   \centering
       \includegraphics[width=0.73\textwidth]{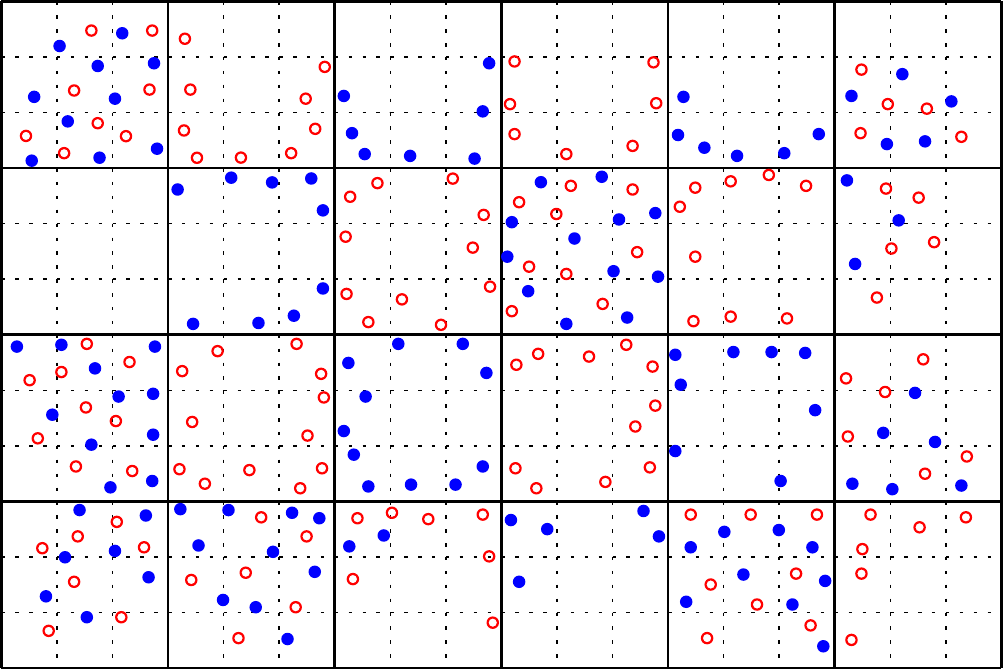}
   \caption{The grid partitioning the points of $P$.}              
   \label{fig:grid1}
\end{figure}

Our algorithm consists of two stages. In Stage~1, we modify the grid cells to satisfy properties (1)-(3), and, in Stage~2, we construct a planar bichromatic spanning tree of the points in each cell and connect between these trees to obtain a planar bichromatic spanning tree of $P$.

\subsection*{Stage~1}
In this stage, we consider the \mono cells and we partition and merge portions of them in order to obtain a subdivision in which all cells are convex and bichromatic. 
Let $C_{i,j}$ be a $3\times 3$ cell of the grid. 
Since $C_{i,j}$ is a $3\times 3$ cell, $C_{i,j}$ is the union of 9 unit sub-cells, labelled $C_{i,j}^k$, for $k=1,2,3,\dots,9$, as shown in Figure~\ref{fig:subcells}(a).
Notice that, since $C_{i,j}$ is a \mono cell, the points of $C_{i,j}$ are of distance at most 1 from the boundary of $C_{i,j}$, and therefore, $C_{i,j}^5$ is empty of points of $P$. 
The region $C_{i,j} \setminus C_{i,j}^5$ is the union of four trapezoids $\T_{i,j}^{l}$, $\T_{i,j}^{r}$, $\T_{i,j}^{t}$, and $\T_{i,j}^{b}$, such that $\T_{i,j}^{l}$ (resp., $\T_{i,j}^{r}$, $\T_{i,j}^{t}$, and $\T_{i,j}^{b}$) is the trapezoid obtained by connecting the left (resp., right, top, and bottom) corners of $C_{i,j}$ by diagonals to the left (resp., right, top, and bottom) corners of $C_{i,j}^5$; see Figure~\ref{fig:subcells}(b).

\begin{figure}[ht]
   \centering
       \includegraphics[width=0.52\textwidth]{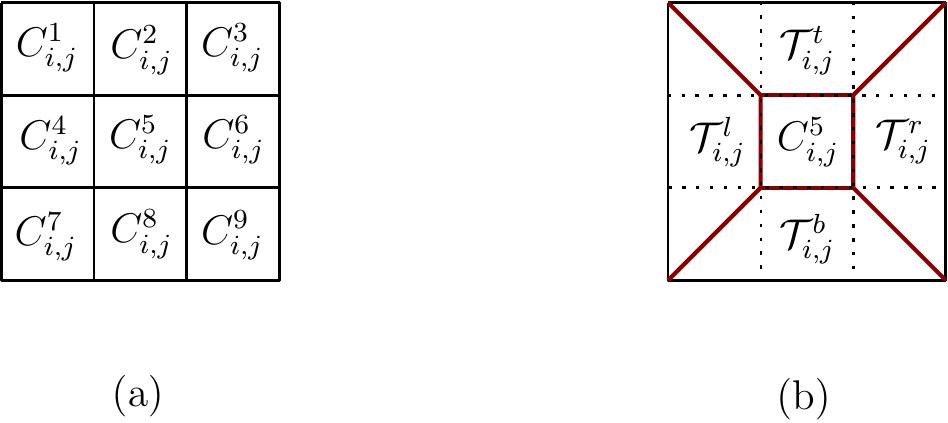}
   \caption{(a) The 9 unit sub-cells of cell $C_{i,j}$. (b) The trapezoids $\T_{i,j}^{l}$, $\T_{i,j}^{r}$, $\T_{i,j}^{t}$, and $\T_{i,j}^{b}$.}              
   \label{fig:subcells}
\end{figure}

\paragraph*{Stage~1.1}
In this stage, we introduce a directed graph $G$ in which the vertices are the \mono cells and the edges are defined as follows.
Let $C_{i,j}$ be a \mono cell, and let
$\N(C_{i,j}) = \{ C_{i,j-1}$, $C_{i,j+1}$, $C_{i-1,j}, C_{i-1,j+1} \}$
be the set of cells that share a grid edge with $C_{i,j}$.
Let $C \in \N(C_{i,j})$ be a \mono cell and assume, w.l.o.g., that $C=C_{i,j+1}$.
There is a directed edge from $C_{i,j}$ to $C_{i,j+1}$ if and only if $C_{i,j}$ and $C_{i,j+1}$ are of different colors and the trapezoid $\T_{i,j}^{r}$ is not empty of points of $P$; see Figure~\ref{fig:directedEdges}.

\begin{figure}[ht]
   \centering
       \includegraphics[width=0.38\textwidth]{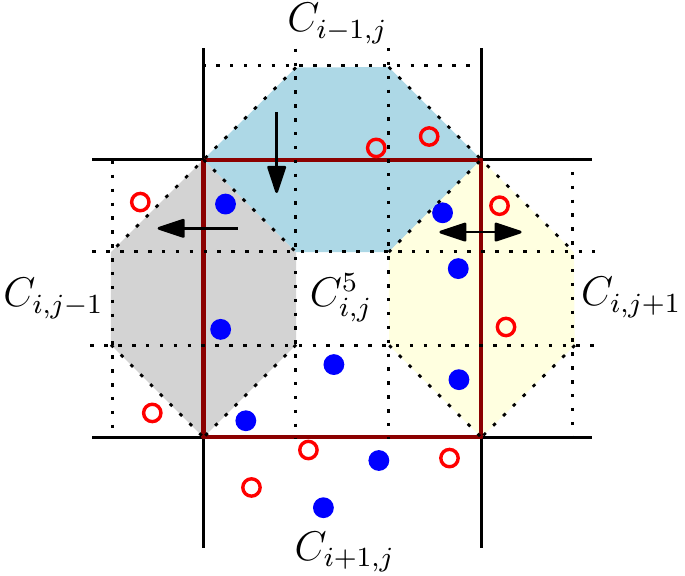}
   \caption{Directed edges between \mono cell $C_{i,j}$ and its \mono neighbors ($C_{i,j-1}$, $C_{i-1,j}$, $C_{i,j+1}$).}              
   \label{fig:directedEdges}
\end{figure}

\vspace{-0.4cm}
\paragraph*{Stage~1.2}
In this stage, we modify the grid cells by partitioning and merging some of the \mono cells with their neighbors, guided by the directed edges introduced in Stage 1.1. 
Before describing how to modify the grid cells, we describe the following \emph{cell partition procedure} that we will apply in this stage to the empty and some of the \mono cells.

\vspace{0.1cm}
\noindent
\textbf{Cell partition procedure.} 
For a \mono cell $C_{i,j}$, partition $C_{i,j} \setminus C_{i,j}^5$ into trapezoids $\T_{i,j}^{l}$, $\T_{i,j}^{r}$, $\T_{i,j}^{t}$, and $\T_{i,j}^{b}$, and merge them with the cells $C_{i,j-1}$ , $C_{i,j+1}$, $C_{i-1,j}$, and $C_{i+1,j}$, respectively; see Figure~\ref{fig:merge1}(b).
\begin{figure}[ht]
	\centering
			\includegraphics[width=0.84\textwidth]{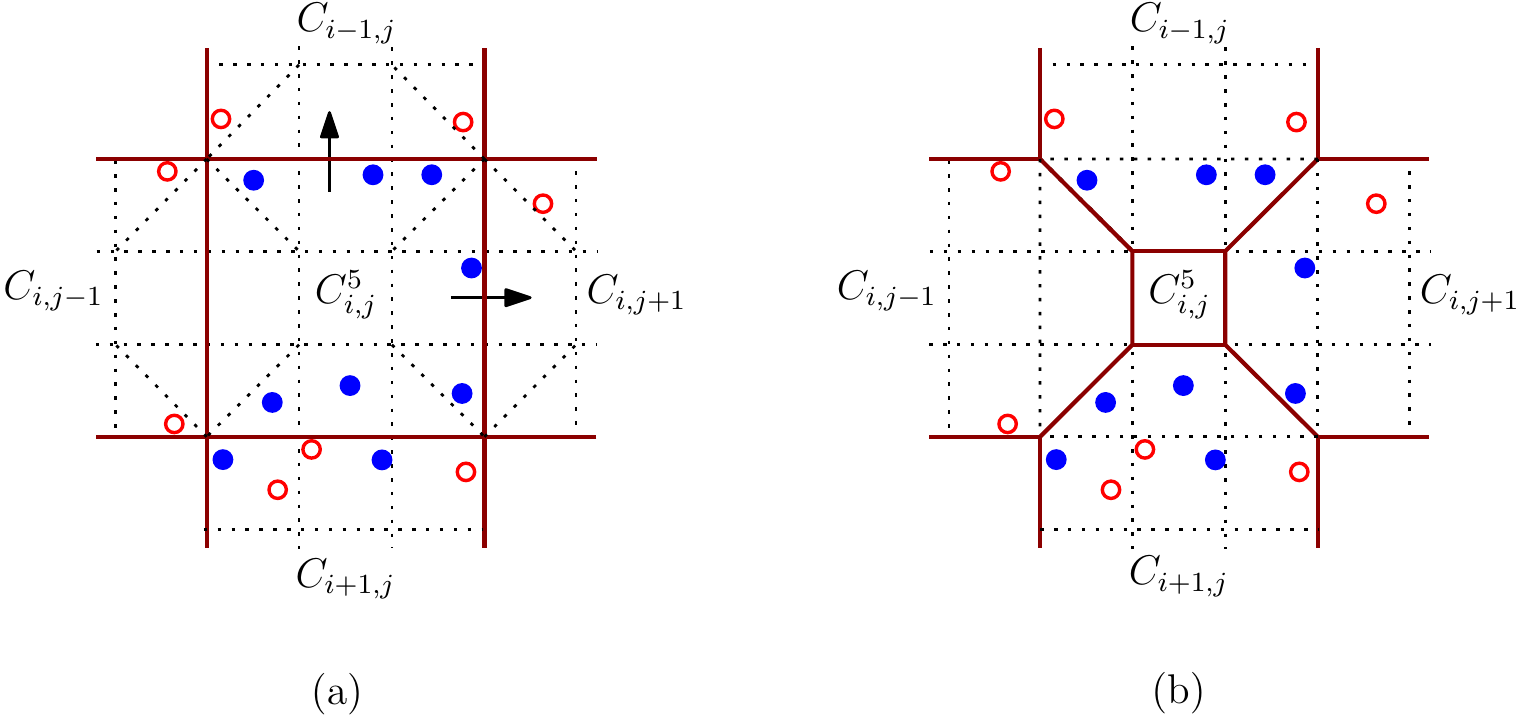}
	\caption{(a) $d_{in}(C_{i,j})=0$ and $d_{out}(C_{i,j})>0$. (b) Partitioning and merging $C_{i,j}$ with its neighbors.}              
	\label{fig:merge1}
\end{figure}

Let $d_{in}(C_{i,j})$ (resp., $d_{out}(C_{i,j})$) denote the in-degree (resp., the out-degree) of the vertex corresponding to the \mono cell $C_{i,j}$ in the graph $G$. We apply the following three steps on the \mono cells. 

\textbf{Step~1.} We apply this step as long as there exists a cell $C_{i,j}$ with $d_{in}(C_{i,j})=0$ and $d_{out}(C_{i,j})>0$. 
For each such cell, we apply the cell partition procedure on $C_{i,j}$ and remove the out-going edges from $C_{i,j}$ and from its neighbors $C_{i,j-1}$, $C_{i,j+1}$, $C_{i-1,j}$, and $C_{i+1,j}$; see Figure~\ref{fig:merge1}.

\textbf{Step~2.} We apply this step on the \mono cells $C_{i,j}$ with $d_{in}(C_{i,j})>0$. Consider the grid as an arbitrary white-black chessboard. For each white cell with $d_{in}(C_{i,j})>0$, we apply the cell partition procedure on $C_{i,j}$.

\textbf{Step~3.} We apply this step on the empty and the \mono cells $C_{i,j}$ with $d_{in}(C_{i,j})=0$ and $d_{out}(C_{i,j})=0$ (that are not considered in the previous steps). For each such cell, we apply the cell partition procedure on $C_{i,j}$.
		
We call a cell $C_{i,j}$ a \emph{partitioned cell} if $C_{i,j}$ has been partitioned (using the cell partition procedure), and we call it an \emph{extended cell} otherwise. If we have two adjacent partitioned cells $C_{i,j}$ and $C_{i,j+1}$, then we call the merged area of the two trapezoids $\T_{i,j}^r$ and $\T_{i,j+1}^l$ a \emph{lune}; see Figure~\ref{fig:lunePartition}. 
\begin{figure}[ht]
	\centering
			\includegraphics[width=0.84\textwidth]{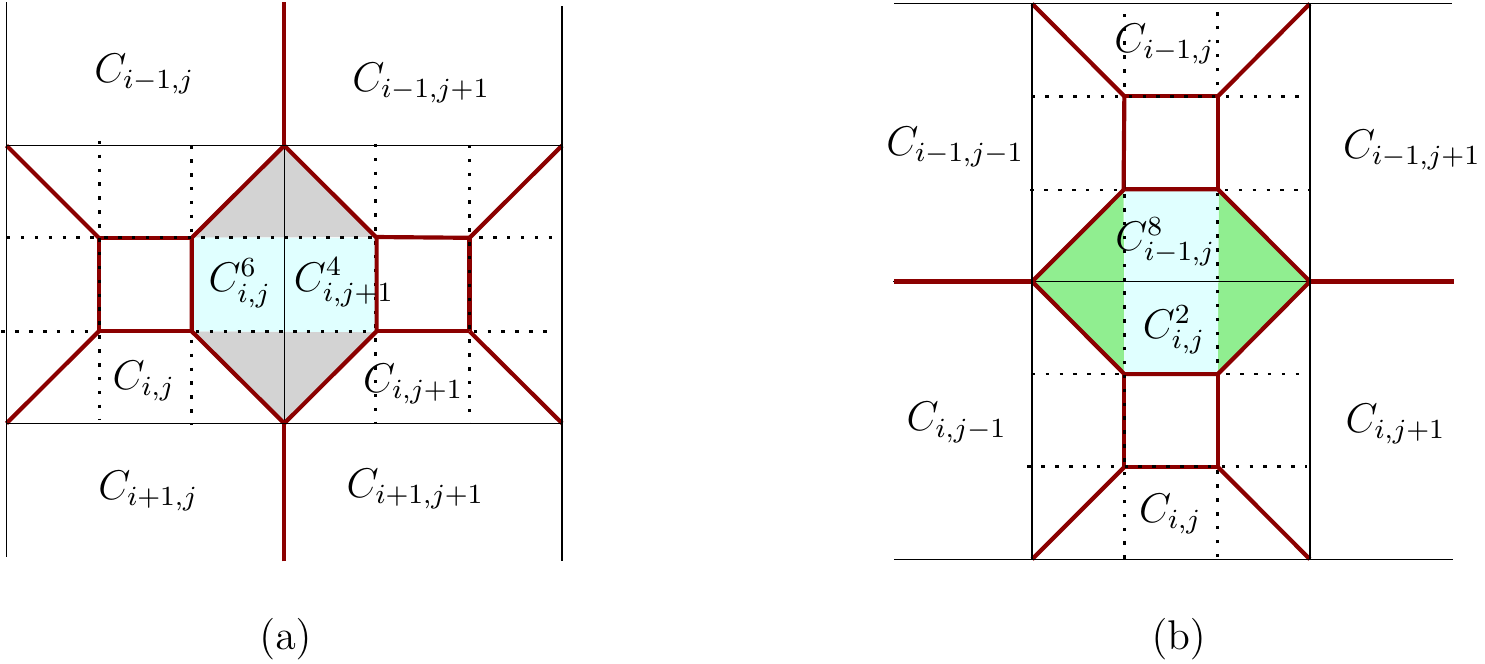}
	\caption{(a) vertical and (b) horizontal lunes.}              
	\label{fig:lunePartition}
\end{figure}
 
At the end of this stage, we have three types of non-empty convex cells: original $3\times 3$ cells, extended cells, and lunes. Clearly, each original cell is bichromatic, otherwise, it would have been partitioned or extended in Steps 1--3. 
Observe that each extended cell $C_{i,j}$ is bichromatic, since $d_{in}(C_{i,j})>0$. 
Observe also that each non-empty lune $L$ is monochromatic. To see this, assume, w.l.o.g., that $L$ was obtained by partitioning the cells $C_{i,j}$ and $C_{i,j+1}$ and merging the trapezoids $\T_{i,j}^r$ and $\T_{i,j+1}^l$. Thus, $L$ cannot be bichromatic, since otherwise, there would be a directed edge from $C_{i,j}$ to $C_{i,j+1}$ and vice versa, which means that one of the cells $C_{i,j}$ and $C_{i,j+1}$ (the black one in the chessboard) is extended in Step~2. 

\paragraph*{Stage~1.3}
In this stage, we get rid of the lunes, by partitioning each lune into sub-pieces and merging the sub-pieces with adjacent extended cells as follows.
Let $L_1$ be a vertical lune obtained by merging two adjacent trapezoids $\T_{i,j}^r$ and $\T_{i,j+1}^l$; see Figure~\ref{fig:lunePartition}(a). (A horizontal lune will be treated analogously.) 
As observed above, $L_1$ is monochromatic (or empty) which means that the subregion $C_{i,j}^6 \cup C_{i,j+1}^4 \subseteq L_1$  is empty of points of~$P$.

We consider the four triangles obtained by removing $C_{i,j}^6$ from $\T_{i,j}^r$ and removing $C_{i,j+1}^4$ from $\T_{i,j+1}^l$, and we merge them with the cells $C_{i-1,j}$ , $C_{i-1,j+1}$, $C_{i+1,j}$, and $C_{i+1,j+1}$ according to the following cases. We describe how to merge the top triangles with $C_{i-1,j}$ and $C_{i-1,j+1}$. (Merging the bottom triangles with $C_{i+1,j}$, and $C_{i+1,j+1}$ is done analogously.)
Let $v_t$ be the top vertex of $L_1$; see Figure~\ref{fig:lunePartition1}.
\begin{itemize}
	\item If both $C_{i-1,j}$ and $C_{i-1,j+1}$ are extended cells, then we merge the top-left triangle with $C_{i-1,j}$ and the top-right triangle with $C_{i-1,j+1}$; see Figure~\ref{fig:lunePartition1}(a).
	\item If $C_{i-1,j}$ is a partitioned cell and $C_{i-1,j+1}$ is an extended cell, then we have another horizontal lune $L_2$ between $C_{i-1,j}$ and $C_{i,j}$; see Figure~\ref{fig:lunePartition1}(b). Notice that the union of the top-left triangle of $L_1$ and the right-bottom triangle of $L_2$ is exactly the sub-cell $C_{i,j}^3$. Notice also that $L_2$ is \mono and has the same color as $L_1$. Thus, the points of $P$ in $C_{i,j}^3$ are of distance 1 from $v_t$. In this case, we merge the top-right triangle of $L_1$ and the right-top triangle of $L_2$ with $C_{i-1,j+1}$. Moreover, we merge the region of $C_{i,j}^3$ intersecting the disk of radius 1 centered at $v_t$ with $C_{i-1,j+1}$; see Figure~\ref{fig:lunePartition1}(b).
\begin{figure}[ht]
	\centering
			\includegraphics[width=0.76\textwidth]{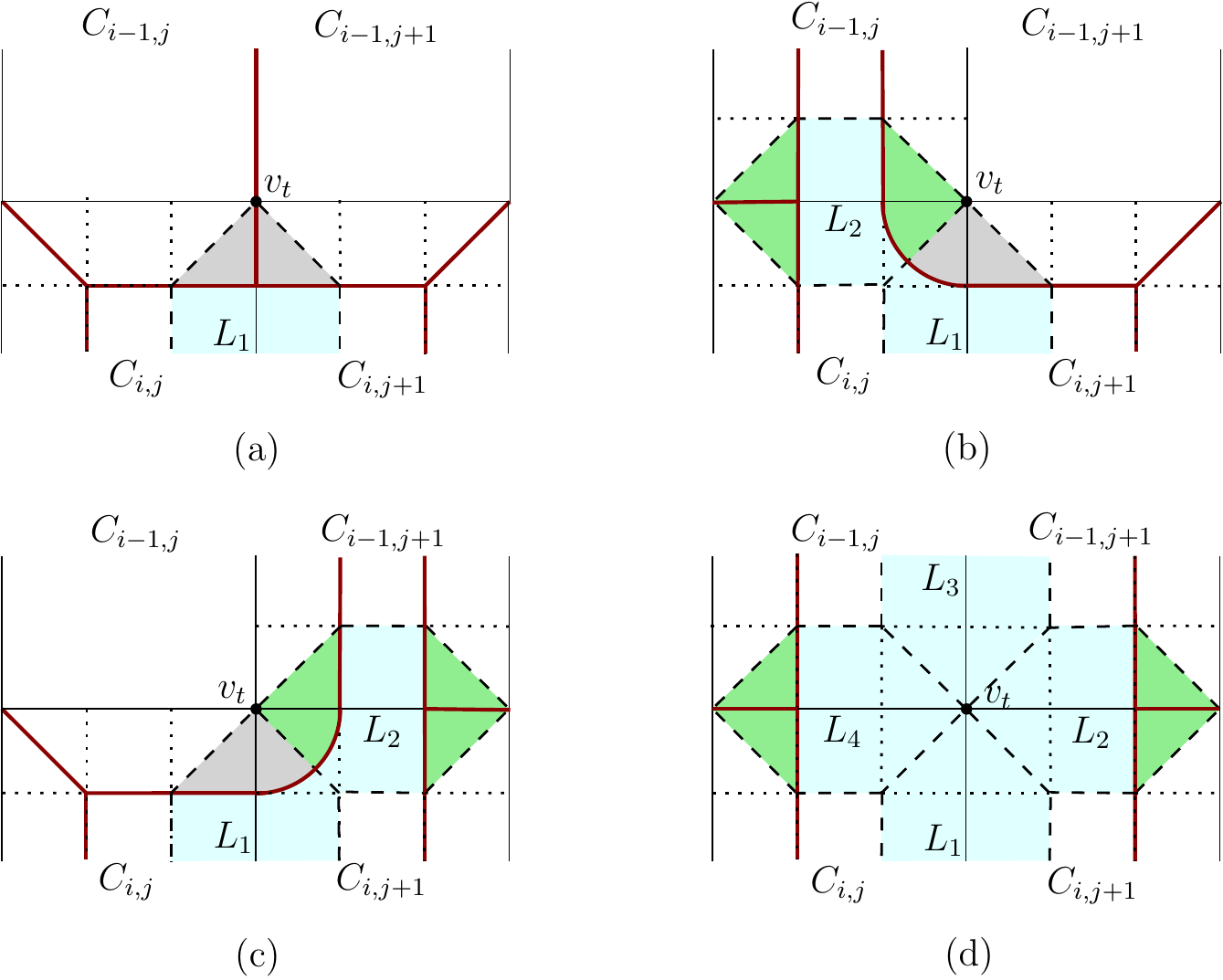}
	\caption{Merging the top triangles of $L_1$ with the cells $C_{i-1,j}$ and $C_{i-1,j+1}$.
	The gray and the green regions are part of vertical and horizontal lunes, respectively, and 
	the light blue regions are empty of points of $P$.
	}              
	\label{fig:lunePartition1}
\end{figure}
\noindent
	\item If $C_{i-1,j}$ is an extended cell and $C_{i-1,j+1}$ is a partitioned cell, then this case is symmetric to the previous case; see Figure~\ref{fig:lunePartition1}(c).
	\item If both $C_{i-1,j}$ and $C_{i-1,j+1}$ are partitioned cells, then we have four lunes incident to $v_t$; see Figure~\ref{fig:lunePartition1}(d). Since all of the lunes are \mono and have the same color, the triangles of these lunes that are incident to $v_t$ are empty of points of $P$ and, therefore, we remove these triangles from the division.
\end{itemize}

Moreover, in each partitioned cell $C_{i,j}$ such that $i=1$, $i=n$, $j=1$, or $j=n$, we treat the trapezoids adjacent to the boundary of the grid as half-lunes and we merge them with their adjacent extended cells as in the lunes case. 

Notice that, at the end of this stage, we have two types of non-empty cells: 
original $3\times 3$ cells and extended cells, and both types are convex and bichromatic cells; see Figure~\ref{fig:bottleneckBound}. From now on, we refer to both types of these cells as extended cells and denote them by $\hat{C}$. That is, $\hat{C}_{i,j}$ is either an original $3\times 3$ cell $C_{i,j}$ or an extended cell obtained by merging $C_{i,j}$ with trapezoids from its neighbors.
\begin{figure}[ht]
	\centering
			\includegraphics[width=0.72\textwidth]{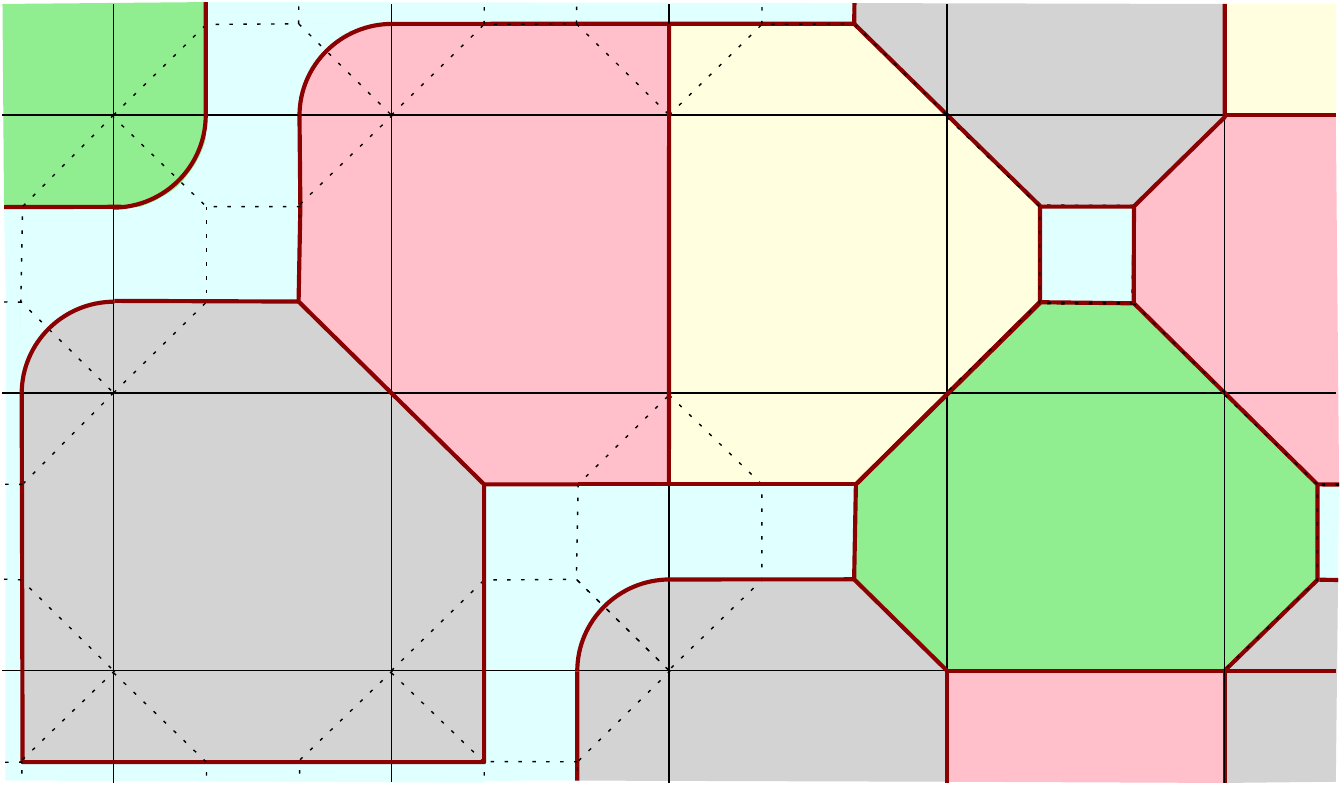}
	\caption{A subdivision obtained at the end of Stage~1. The light blue regions are empty of points of $P$.} 	
	\label{fig:bottleneckBound}
\end{figure}

\subsection*{Stage~2}
In this stage, we construct a planar bichromatic spanning tree in each (extended) cell and connect them to each other to obtain, overall, a planar bichromatic spanning tree of $P$.
For each cell $\hat{C}_{i,j}$, we denote by $\hat{P}_{i,j}$ the set of points of $P$ lying in $\hat{C}_{i,j}$.
If $C_{i,j}$ has been partitioned, then we set $\hat{P}_{i,j} = \emptyset$.

\vspace{-0.2cm}
\paragraph*{Stage~2.1}
\vspace{-0.1cm}
In each cell $\hat{C}_{i,j}$, we construct a planar bichromatic spanning tree $T_{i,j}$ of $\hat{P}_{i,j}$ as follows. We select an arbitrary red point $s \in \hat{P}_{i,j}$ as a center of the tree and connect it to each blue point in the cell to produce a star. We extend the edges of the star to partition the cell into convex cones, possibly except one cone; see Figure~\ref{fig:star}. If we have a non-convex cone, then we divide it into two convex cones by adding its bisector, as shown in Figure~\ref{fig:star}(right).  
Then, we connect all the red points in each cone to one of the blue points on the lines bounding the cone.
\begin{figure}[ht]
   \centering
       \includegraphics[width=0.56\textwidth]{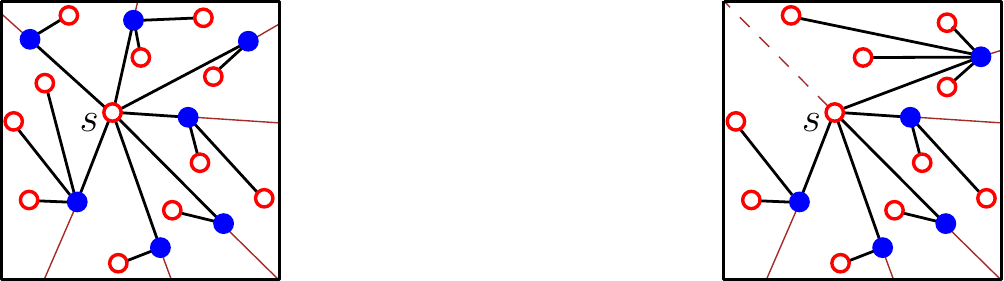}
   \caption{Constructing a planar bichromatic spanning tree in a cell.}            
	\label{fig:star}
\end{figure}
%
\begin {lemma} \label{lemma:connecting point to star}
Let $T_{i,j}$ be a tree constructed in  Stage~2.1 in cell $\hat{C}_{i,j}$ . 
Any (red or blue) point $p$ in the plane can be connected to $T_{i,j}$ without crossing the edges of $T_{i,j}$.
\end{lemma}

\begin{proof}
Let $s$ be the center of $T_{i,j}$ and recall that its color is red. Consider the cones produced by the rays between $s$ and the blue points of $T_{i,j}$. Let $C$ be the cone containing $p$ and let $a$ and $b$ be the two blue points defining $C$. By the way we constructed $T_{i,j}$, all the points in $C$ are red and connected to exactly one of the points $a$ and $b$, assume, w.l.o.g., $a$. We distinguish between two cases with respect to the color of $p$. \\
\textbf{Case~1:} $p$ is a blue point. If the edge $(s,p)$ does not cross the edges of $T_{i,j}$, then we connect $p$ to $s$. Otherwise, we connect $p$ to the endpoint of the first edge (from $p$) crossing $(s,p)$; see Figure~\ref{fig:pointToStar}(a). \\
\textbf{Case~2:} $p$ is a red point. In this case, we connect $p$ to $a$; see Figure~\ref{fig:pointToStar}(b).
\vspace{-0.1cm} 
\end{proof}
\begin{figure}[ht]
	\centering
			\includegraphics[width=0.76\textwidth]{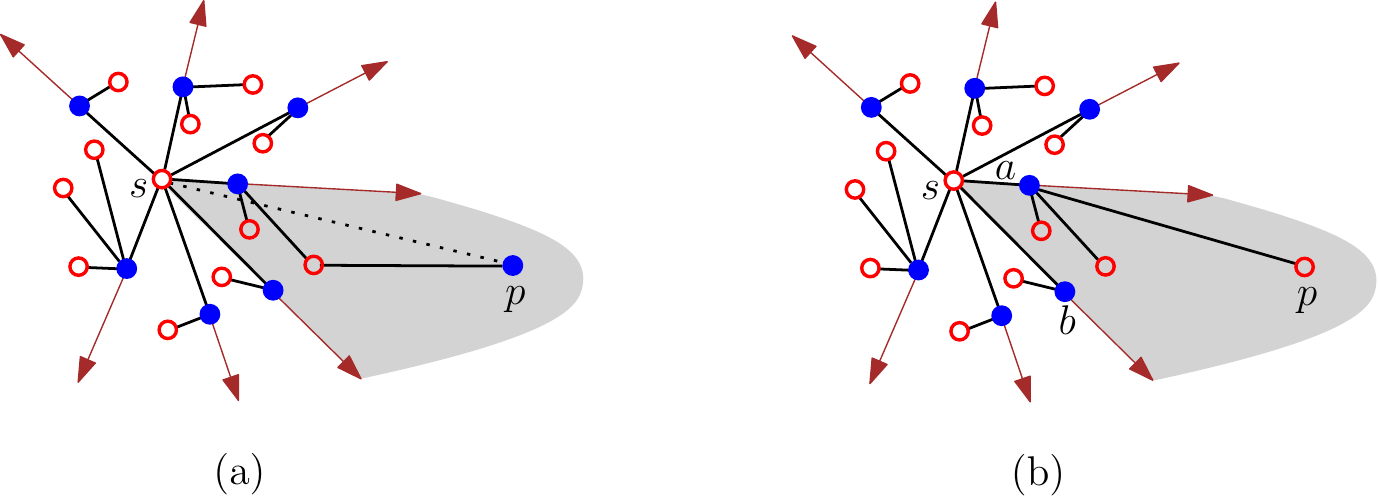}
	\caption{(a) $p$ and $s$ are of different colors. (b) $p$ and $s$ are of the same color.} 	
	\label{fig:pointToStar}
\end{figure}


\vspace{-0.4cm}
\paragraph*{Stage~2.2}

In this stage, we connect between the trees that are constructed in Stage~2.1 to obtain a planar bichromatic spanning tree of $P$.
Let $\hat{C}_{i,j}$ and $\hat{C}_{k,l}$ be two (extended) cells. We say that $\hat{C}_{k,l}$ is a \emph{side} adjacent (or s-adjacent for short) cell of $\hat{C}_{i,j}$, if one of the following holds:
\begin{itemize}
	\item $k=i+1$ and $l=j$, or
	\item $k=i$ and $l=j+1$,
\end{itemize}
and we say that $\hat{C}_{k,l}$ is a \emph{diagonal} adjacent (or d-adjacent for short) cell of $\hat{C}_{i,j}$, if one of the following holds:
\begin{itemize}
	\item $k=i-1$, $l=j+1$, and $C_{i-1,j}$ and $C_{i,j+1}$ have been partitioned, or
	\item $k=i-1$, $l=j-1$, and $C_{i,j-1}$ and $C_{i-1,j}$ have been partitioned.
\end{itemize}

We construct a bichromatic spanning tree $T'$ of $P$ by traversing the cells starting from an arbitrary (non-empty) cell  (in breadth first search (BFS) manner).
That is, we first initiate a tree $T'$ by an arbitrary tree $T_{i,j}$ that is constructed in a cell $\hat{C}_{i,j}$. Then, we connect $T'$ to all the trees constructed in the cells adjacent to $\hat{C}_{i,j}$, and proceed from these trees. 
More precisely, in each step, we consider a tree $T_{i,j}$, which is already connected to $T'$, and we connect $T'$ to all of the trees constructed in the cells adjacent to $\hat{C}_{i,j}$ via $T_{i,j}$ (if they are not connected yet to $T'$). 
In the following, we describe how to connect $T'$ to all the trees constructed in the cells adjacent to $\hat{C}_{i,j}$.

\old{ 
Recall that $\N(C_{i,j})$ denotes the set $\{ C_{i,j-1}$, $C_{i,j+1}$, $C_{i-1,j}, C_{i-1,j+1} \}$ of (original) cells that are adjacent to the (original) cell $C_{i,j}$.
For each cell $C \in \N(C_{i,j})$, if $C$ has not been partitioned and the tree $T_C$ constructed in $C$ is not connected yet to $T'$, then we connect $T_{i,j}$ to $T_C$. 
Moreover, if two cells in $\N(C_{i,j})$ that are incident to the same corner are partitioned cells, 
for example $C_{i-1,j}$ and $C_{i,j+1}$, then, if $C_{i-1,j+1}$ is not partitioned and not connected yet to $T'$, then we connect $T_{i,j}$ to the tree constructed in $\hat{C}_{i-1,j+1}$.

Let $v_{tr}$, $v_{tl}$, $v_{br}$, and $v_{bl}$ be the top-right, top-left, bottom-right, and bottom-left vertices of the grid incident to $C_{i,j}$, respectively; see Figure~\ref{fig:case1}.
Let $v$ be a vertex of the grid incident to $C_{i,j}$ and assume, w.l.o.g., that $v=v_{tr}$.
We distinguish between two cases.
In Case~1, we show how to connect two trees constructed in extended cells that share a diagonal edge and, in Case~2, we show how to connect two trees constructed in cells that share a vertical or a horizontal edge.
More precisely, for each tree $T_{i,j}$, in Case~1, we show how to connect $T_{i,j}$ to the tree $T_{i-1,j+1}$ constructed in $\hat{C}_{i-1,j+1}$ (in case that $C_{i-1,j}$ and $C_{i,j+1}$ are partitioned cells), and, in Case~2, for each $C \in \N(C_{i,j})$, we show how to connect $T_{i,j}$ to $T_C$ (in case that $C$ is not partitioned and the tree $T_C$ constructed in the extended cell $\hat{C}$ is not connected yet to $T'$).
}

Let $\hat{C}$ be a cell adjacent to $\hat{C}_{i,j}$, such that the tree $T_C$ constructed in $\hat{C}$ is not connected yet to $T'$. Let $v_{tr}$, $v_{tl}$, $v_{br}$, and $v_{bl}$ be the top-right, top-left, bottom-right, and bottom-left vertices of the grid incident to $C_{i,j}$, respectively; see Figure~\ref{fig:case1}. We distinguish between two cases. \\
\textbf{Case~1:} $\hat{C}$ is a d-adjacent cell of $\hat{C}_{i,j}$. Assume, w.l.o.g., that $\hat{C} = \hat{C}_{i-1,j+1}$. Then, the boundaries of $\hat{C}_{i,j}$ and $\hat{C}_{i-1,j+1}$ share a common (diagonal) edge $\overline{ab}$; see Figure~\ref{fig:case1}. 
Moreover, the convex hull of $\hat{C}_{i,j} \cup \hat{C}_{i-1,j+1}$ does not contain any point of $P\setminus (\hat{P}_{i,j} \cup \hat{P}_{i-1,j+1})$ (this can be seen clearly in Figure~\ref{fig:bottleneckBound}).
Let $p \in \hat{P}_{i,j}$ be the closest point to the line passing through $\overline{ab}$, such that no edge of $T'$ crosses the triangle $\Delta pab$. By Claim~\ref{claim:empty triangle}, such a point $p$ exists.
Then, any edge connecting $p$ to any point of $T_{i-1,j+1}$ does not cross any non-empty cell except $\hat{C}_{i,j}$ and $\hat{C}_{i-1,j+1}$. 
Therefore, by Lemma~\ref{lemma:connecting point to star}, we can connect $p$ to $T_{i-1,j+1}$ without crossing any other edge of $T'$.

\begin{figure}[ht]
	\centering
			\includegraphics[width=0.32\textwidth]{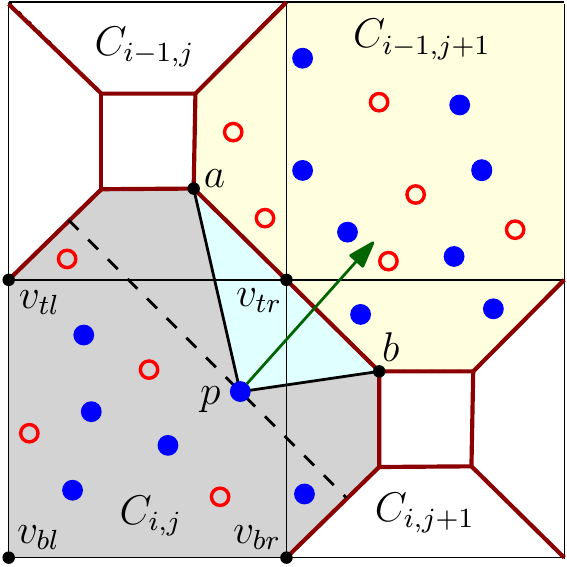}
	\caption{Illustration of Case~1. We connect $T_{i-1,j+1}$ to $T_{i,j}$ via $p$.} 	
	\label{fig:case1}
\end{figure}

\noindent
\textbf{Case~2:} $\hat{C}$ is an s-adjacent cell of $\hat{C}_{i,j}$. Assume, w.l.o.g., that $\hat{C} = \hat{C}_{i,j+1}$.
Let $p$ be the rightmost point in $\hat{P}_{i,j}$, such that no edge of $T'$ crosses the triangle $\Delta pv_{br}v_{tr}$; 
see Figure~\ref{fig:case21}. By Claim~\ref{claim:empty triangle}, such a point $p$ exists. 
Let $H$ be the convex hull of $\hat{P}_{i,j+1} \cup \{p\}$.
We consider two sub-cases. \\
\textbf{Case~2.1:}  $H \cap (P \setminus (\hat{P}_{i,j} \cup \hat{P}_{i,j+1}))= \emptyset $  (i.e., $H$ does not contain any point of $P$ that is not in $\hat{P}_{i,j}  \cup \hat{P}_{i,j+1}$); see Figure~\ref{fig:case21}. 
  Therefore, by Lemma~\ref{lemma:connecting point to star}, we can connect $T_{i,j+1}$ to $T_{i,j}$ via $p$, without crossing any other edge of $T'$. 
	
\begin{figure}[ht]
	\centering
			\includegraphics[width=0.78\textwidth]{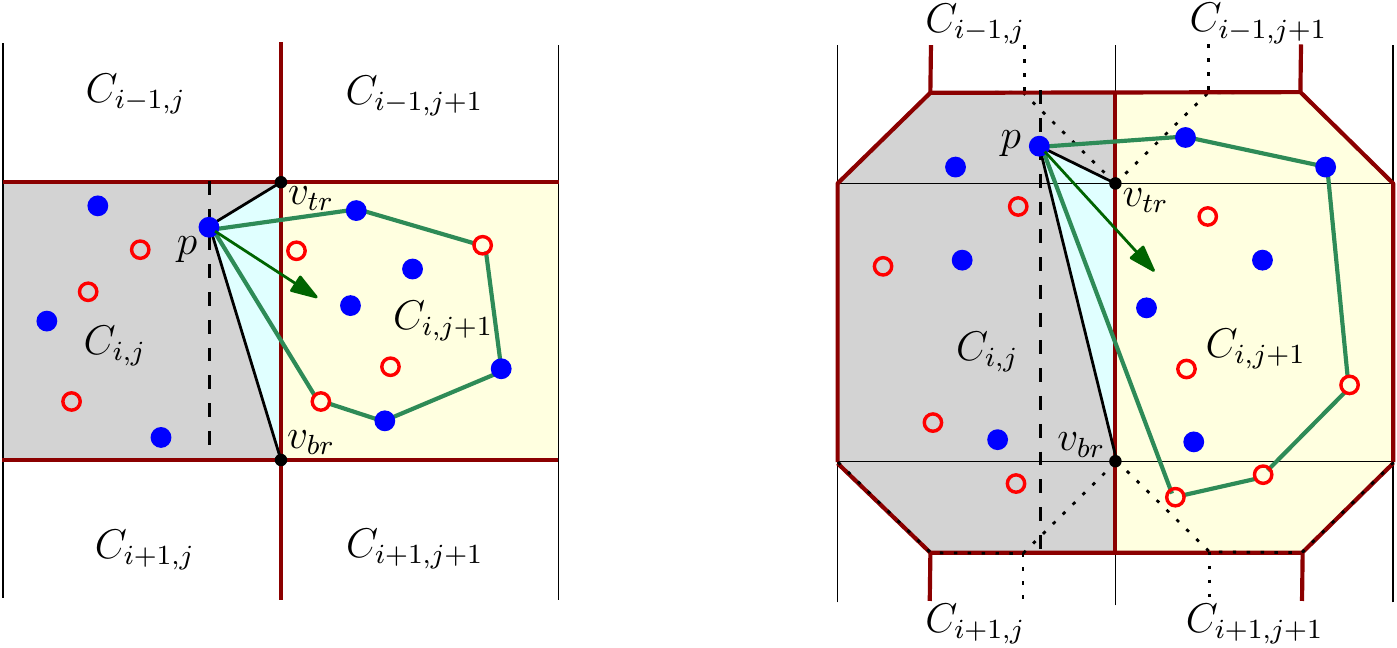}
	\caption{The convex hull $H$ of $\hat{P}_{i,j+1} \cup \{p\}$ (consisting of green segments) does not contain any point of $P$ that is not in $\hat{P}_{i,j}  \cup \hat{P}_{i,j+1}$. We connect $T_{i,j+1}$ to $T_{i,j}$ via $p$.} 	
	\label{fig:case21}
\end{figure}

\noindent
\textbf{Case~2.2:}  $H \cap (P \setminus (\hat{P}_{i,j} \cup \hat{P}_{i,j+1})) \neq \emptyset $  (i.e., $H$  contains a point of $P$ that is not in  $\hat{P}_{i,j}  \cup  \hat{P}_{i,j+1}$). 
In this case, $H$ contains a point in $\hat{P}_{i-1,j} \cup \hat{P}_{i-1,j+1}$ or in $\hat{P}_{i+1,j} \cup \hat{P}_{i+1,j+1}$.
Assume, w.l.o.g., that $H$ contains a point in $\hat{P}_{i-1,j} \cup \hat{P}_{i-1,j+1}$; see Figure~\ref{fig:case22}.
Notice that exactly one of the sets $\hat{P}_{i-1,j}$ or $\hat{P}_{i-1,j+1}$ is an empty set, 
since, in this case, exactly one of the cells $C_{i-1,j}$  or $C_{i-1,j+1}$ has been partitioned.
We further distinguish between two cases.
\begin{enumerate}
   \item $H \cap \hat{P}_{i-1,j} \neq \emptyset$; see Figure~\ref{fig:case22}(a).
   In this case we first connect $T_{i,j+1}$ to $T_{i-1,j}$  as follows.
   Let $q \in \hat{P}_{i-1,j}$ be the closest point to the line passing through the boundary edge between $\hat{C}_{i-1,j}$ and $\hat{C}_{i,j+1}$.
	Then, the convex hull of $\hat{P}_{i,j+1} \cup \{ q \}$ does not contain any point of $P$ that is not in $\hat{P}_{i,j+1} \cup \{ q \}$.
	Therefore, by Lemma~\ref{lemma:connecting point to star}, we can connect $T_{i,j+1}$ to $T_{i-1,j}$ via $q$, without crossing any other edge of $T'$.
	
	Moreover, if $T_{i-1,j}$ is not connected yet to $T'$, then we apply Case~2 on $\hat{C}_{i-1,j}$ to connect $T_{i-1,j}$ to $T_{i,j}$.

\begin{figure}[ht]
	\centering
			\includegraphics[width=0.78\textwidth]{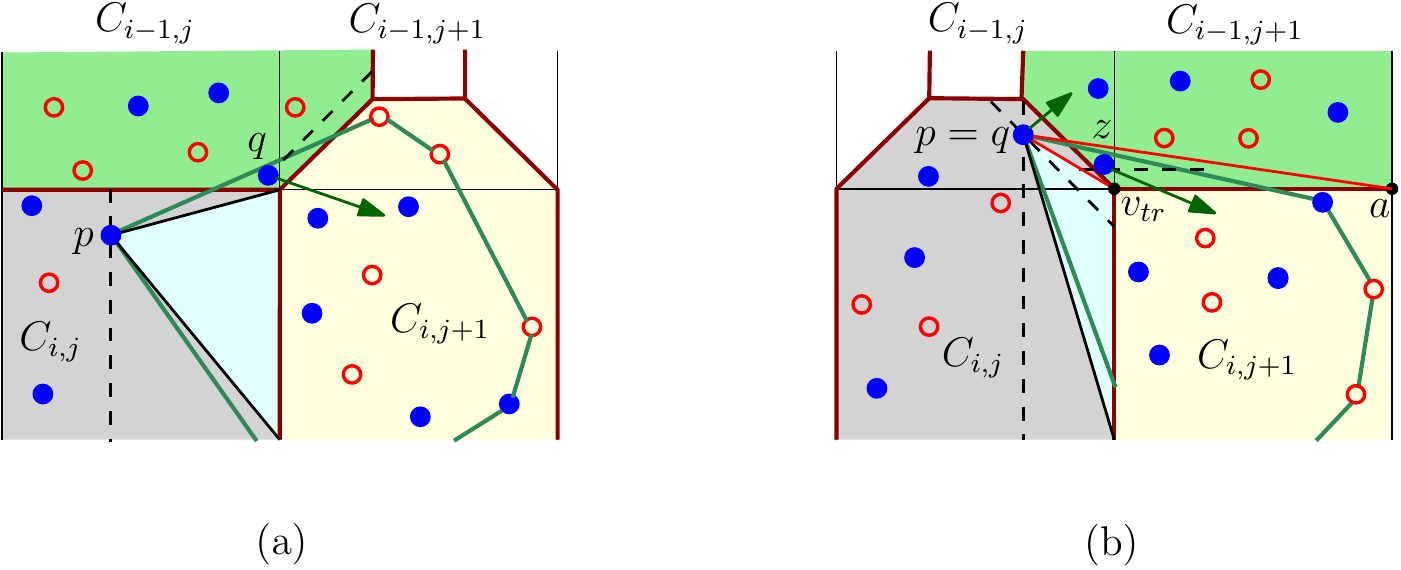}
	\caption{(a) $H$ contains points from $\hat{P}_{i-1,j}$. We connect $T_{i,j+1}$ to $T_{i-1,j}$ via $q$. (b) $H$ contains a point $z$ from $\hat{P}_{i-1,j+1}$. We connect $T_{i-1,j+1}$ to $T_{i,j}$ via $q$ and we connect $T_{i,j+1}$ to $T_{i-1,j+1}$ via $z$.} 	
	\label{fig:case22}
\end{figure}

   \item $H \cap \hat{P}_{i-1,j+1} \neq \emptyset$.
	If $T_{i-1,j+1}$ is not connected yet to $T'$, then we first connect $T_{i-1,j+1}$ to $T_{i,j}$ as follows.
  Let $q \in \hat{P}_{i,j}$ be the closest point to the line passing through the boundary edge between $\hat{C}_{i,j}$ and $\hat{C}_{i-1,j+1}$; 
	see Figure~\ref{fig:case22}(b). 
	Then, the convex hull of $\hat{P}_{i-1,j+1} \cup \{ q \}$ does not contain any point of $P$ that is not in $\hat{P}_{i-1,j+1} \cup \{ q \}$.
	Therefore, by Lemma~\ref{lemma:connecting point to star}, we can connect $T_{i-1,j+1}$ to $T_{i,j}$ via $q$, without crossing any other edge of $T'$.

Moreover, we connect $T_{i,j+1}$ to $T_{i-1,j+1}$ as follows. Let $a$ be the bottom-right corner of $C_{i-1,j+1}$. Thus, $\Delta{pv_{tr}a} \cap \hat{P}_{i-1,j+1} \neq \emptyset$ and $\Delta{pv_{tr}a} \cap (\hat{P}_{i,j}\setminus\{p\}) = \emptyset$. Let $z$ be the bottommost point in $\Delta{pv_{tr}a} \cap \hat{P}_{i-1,j+1}$, such that no edge of $T'$ crosses the triangle $\Delta zv_{tr}a$; see Figure~\ref{fig:case22}(b). By Claim~\ref{claim:empty triangle}, such a point $z$ exists.
Then, the convex hull of $\hat{P}_{i,j+1} \cup \{ z \}$ does not contain any point of $P$ that is not in $\hat{P}_{i,j+1} \cup \{ z \}$. Therefore, by Lemma~\ref{lemma:connecting point to star}, we can connect $T_{i,j+1}$ to $T_{i-1,j+1}$ via $z$, without crossing any other edge of $T'$. (Notice that the two edges added in this case do not cross each other.)

\end{enumerate}

\subsubsection*{Correctness Proof}
Recall that $T$ is a bichromatic spanning tree of $P$ of minimum bottleneck $\lambda$. In this section, we prove that $T'$ is a planar bichromatic spanning tree of $P$ of bottleneck at most $8\sqrt{2}\lambda$. Notice that every point $p \in P$ is contained in a bichromatic cell $\hat{C}_{i,j}$ and hence, it is connected to $T_{i,j}$, the tree constructed in Stage~2.1 in $\hat{C}_{i,j}$. Therefore, to show that $T'$ is a bichromatic spanning tree of $P$, it is sufficient to show that each tree $T_{i,j}$ is connected to $T'$.

\begin{claim} \label{claim:empty triangle}
Let $T'$ be the tree constructed at some step during Stage~2.2 and assume that $T'$ is planar.
Let $T_{i,j}$ be a tree constructed in $\hat{C}_{i,j}$ and $T_{i,j}$ is already connected to $T'$.  
Let $\hat{C}$ be an adjacent cell of $\hat{C}_{i,j}$ that shares an edge $\overline{ab}$ with $\hat{C}_{i,j}$ and let $T_C$ be the tree constructed in $\hat{C}$, and assume that $T_C$ is not connected yet to $T'$.
Then, there exists a point $p$ in $\hat{C}_{i,j}$, such that no edge of $T'$ crosses the triangle $\Delta p a b$.
\end{claim}
\begin{proof}
Assume, w.l.o.g., that $\hat{C} = \hat{C}_{i,j+1}$, $a=v_{tr}$, and $b=v_{br}$; see Figure~\ref{fig:emptyTriangle}.
The following procedure shows the existence of such a point $p$.
We sweep leftwards in $\hat{C}_{i,j}$ with a vertical line $l$, starting from $\overline{v_{tr} v_{br}}$ until we meet a point, or an edge of $T'$. If we first meet a point, then this point satisfies the claim. 
Otherwise, we first meet an edge $(p',q')$ of $T'$; see Figure~\ref{fig:emptyTriangle}. This could only be when exactly one of the endpoints $p'$ or $q'$ is outside $\hat{C}_{i,j}$. 
Let $C_l$ and $C_r$ be the two sub-cells obtained by partitioning $\hat{C}_{i,j}$ with the line that goes through the points $p'$ and $q'$. Let $C_r$ be the sub-cell containing $v_{tr}$ and $v_{br}$. We keep sweeping leftwards only inside $C_r$. 
As before, if we first meet a point, then this point satisfies the claim. 
Otherwise, we meet an edge $(p'',q'')$ of $T'$ before we meet a point. Then, one of the endpoints $p''$ or $q''$ 
is outside $\hat{C}_{i,j}$. Let $C_{rl}$ and $C_{rr}$ be the two sub-cells obtained by partitioning $C_r$ with the line that goes through the points $p''$ and $q''$. Let $C_{rr}$ be the sub-cell containing $v_{tr}$ and $v_{br}$. We keep sweeping leftwards only inside $C_{rr}$, until we meet a point, and this point satisfies the claim. 
Notice that, in the last sweep we meet a point before we meet an edge of $T'$.
This follows from the planarity of $T'$.
\end{proof} 

\begin{figure}[ht]
	\centering
			\includegraphics[width=0.31\textwidth]{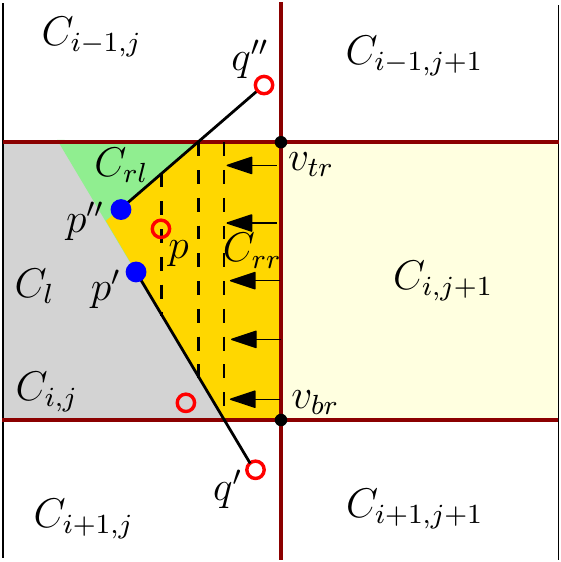}
	\caption{Illustration of the proof of Claim~\ref{claim:empty triangle}.} 	
	\label{fig:emptyTriangle}
\end{figure}

\begin{lemma} \label{lemma:adjacent cells connected}
Let $T_{i,j}$ be a tree constructed in $\hat{C}_{i,j}$ in Stage~2.1 and assume that $T_{i,j}$ is already connected to $T'$. Then, at the end of Stage~2.2, all the trees that are constructed in the  cells adjacent to $\hat{C}_{i,j}$ are connected to $T'$ as well.  
\end{lemma}

\begin{proof}
Let $\hat{C}$ be an adjacent cell of $\hat{C}_{i,j}$.
We distinguish between two cases. \\
\textbf{Case~1:} $\hat{C}$ is a d-adjacent cell of $\hat{C}_{i,j}$. Assume, w.l.o.g., that $\hat{C} = \hat{C}_{i-1,j+1}$.
  Then, in Stage~2.2, Case 1, we connect between $T_{i,j}$ and $T_{i-1,j+1}$.\\
\textbf{Case~2:} $\hat{C}$ is an s-adjacent cell of $\hat{C}_{i,j}$. Assume, w.l.o.g., that $\hat{C} = \hat{C}_{i,j+1}$.
As described in Stage~2.2, we select a point $p \in \hat{C}_{i,j}$ and compute the convex hull $H$ of  $ \{p\}  \cup \hat{P}_{i,j+1}$. 
Then, we consider two cases.
In Case~2.1, when $H$ does not contain any point of $P \setminus (\hat{P}_{i,j}  \cup \hat{P}_{i,j+1}) $, 
we connect $T_{i,j}$ directly to $T_{i,j+1}$ (via $p$).
And, in Case~2.2, when $H$ contains a point of $P \setminus (\hat{P}_{i,j}  \cup \hat{P}_{i,j+1})$, we connect $T_{i,j+1}$ to $T_{i,j}$ via the tree $T_{i-1,j+1}$, in case that $H$ contains a point of $\hat{P}_{i-1,j+1}$ (or via the tree $T_{i+1,j+1}$, in case that $H$ contains a point of $\hat{P}_{i+1,j+1}$).\\
In the case that $H$ contains a point of $\hat{P}_{i-1,j}$ (symmetrically, $H$ contains a point of $\hat{P}_{i+1,j}$), we connect $T_{i,j+1}$ to $T_{i,j}$ via the tree $T_{i-1,j}$. 
If $T_{i-1,j}$ is already connected to $T_{i,j}$, then we are done. 
Otherwise, since $\hat{C}_{i-1,j}$ is an s-adjacent cell of $\hat{C}_{i,j}$, we will try to connect $T_{i-1,j}$ to $T_{i,j}$ in the next iteration in Stage~2.2 (by applying Case~2 once again). 
In the next iteration, either we connect $T_{i-1,j}$ to $T_{i,j}$ in one of the cases described above or we end up by connecting $T_{i,j-1}$ to $T_{i-1,j}$. 
In the latter case, if $T_{i,j-1}$ is already connected to $T_{i,j}$, then we are done. 
otherwise, $T_{i+1,j}$ is already connected to $T_{i,j}$. 
In this case, we connect $T_{i,j-1}$ to $T_{i,j}$ either directly or via $T_{i+1,j}$, and we are done.
\end{proof}

\begin{lemma} \label{lemma:two points connected}
Let $p$ and $q$ be two points of $P$, such that $p$ and $q$ are of different colors, $|pq| \le \lambda$ and $p$ belongs to $T'$. Then, $q$ also belongs to $T'$.
\end{lemma}

\begin{proof}
Since $|pq| \le \lambda$, either $p$ and $q$ are in the same cell or they are in adjacent cells. 
If they are in the same cell $\hat{C}_{i,j}$, then, after Stage~2.1, they are connected in $T_{i,j}$, and the lemma holds.
Otherwise, assume, w.l.o.g., that $p \in \hat{C}_{i,j}$ and $q \in \hat{C}$, where $\hat{C}$ is adjacent to $\hat{C}_{i,j}$. 
Then, after Stage~2.1, $p$ belongs to $T_{i,j}$ and $q$ belongs to $T_C$, the tree constructed in $\hat{C}$.
Since $T_{i,j}$ is part of $T'$ and $\hat{C}$ is adjacent to $\hat{C}_{i,j}$, by Lemma~\ref{lemma:adjacent cells connected}, $T_{i,j}$ is connected to $T_C$ and therefore $q$ belongs to $T'$.
\end{proof} 

\begin{lemma} \label{lemma:tree connected}
Let $T_{i,j}$ be a tree constructed in $\hat{C}_{i,j}$. Then, $T_{i,j}$ is connected to $T'$. 
\end{lemma}

\begin{proof}
Assume by contradiction that $T_{i,j}$ is not connected to $T'$. Let $a$ be a point from $T'$ and let $b$ be a point from $T_{i,j}$. Since $T$ is a bottleneck bichromatic spanning tree of $P$, there is a path $\Pi$ between $a$ and $b$ in $T$. Let $p$ be the last point (from $a$) on $\Pi$ that belongs to $T'$, i.e., no point of $T'$ appears on the sub-path of $\Pi$ between $p$ and $b$. Since $b$ does not belong to $T'$, such a point $p$ exists.
Let $q$ be the point between $p$ and $b$ on $\Pi$ that is connected to $p$. By the selection of $p$, $q$ does not belong to $T'$. 
Since the bottleneck of $T$ is $\lambda$, we have $|pq| \le \lambda$. Therefore, by Lemma~\ref{lemma:two points connected}, $p$ and $q$ are connected in $T'$, which contradicts that $q$ does not belong to $T'$.    
\end{proof}

\begin{lemma} \label{lemma:planar}
$T'$ is planar. 
\end{lemma}
\begin{proof}
Each $T_{i,j}$ is planar. 
We start with $T'=T_{i,j}$, where $T_{i,j}$ is an arbitrary tree constructed in $\hat{C}_{i,j}$, and in each step, we extend $T'$ by connecting it to the trees corresponding to the cells adjacent to the current cell.    
We connect $T'$ to a ``new'' tree $T_{i,j}$ by picking a point $p$ in $T'$, such that the convex hull $H$ of $\{p\} \cup \hat{P}_{i,j}$ is empty of any other points and no edge of $T'$ crosses $H$. In Claim~\ref{claim:empty triangle}, we showed that such a point $p$ always exists.
Thus, connecting $p$ to any point of $T_{i,j}$ will not cross any other edge of $T'$ nor of any other tree.
Moreover, in Lemma~\ref{lemma:connecting point to star}, we show that we can always connect $p$ to $T_{i,j}$ without crossing any of the edges of $T_{i,j}$. 
Therefore, connecting $T'$ to $T_{i,j}$ does not produce any crossing. 
\end{proof} 

\begin{lemma} \label{lemma:bottleneck bound}
The bottleneck of $T'$ is at most $8\sqrt{2}\lambda$. 
\end{lemma}

\begin{proof}
Consider Figure~\ref{fig:bottleneckBound}. After Stage~1, each extended cell is contained in a square of size $5\lambda \times 5\lambda$, and hence the bottleneck of each tree constructed in Stage~2.1 is at most $5\sqrt{2}\lambda$. Moreover, every two d-adjacent cells are contained in a square of size $8\lambda \times 8\lambda$  and every two s-adjacent cells are contained in a square of size $8\lambda \times 5\lambda$. Thus, each edge added in Stage~2.2 is of length at most $8\sqrt{2}\lambda$. Therefore, each edge in $T'$ is of length at most $8\sqrt{2}\lambda$.  
\end{proof}

\old{
\begin{lemma} \label{lemma:polynomial}
The running time of the algorithm is polynomial. 
\end{lemma}

\begin{proof}
Computing $\lambda$ takes $O(n\log{n})$ time using the algorithm of Biniaz et al.~\cite{Biniaz18}. 

\end{proof}
}

The algorithm consists of two main stages, and each one of them can be implemented in polynomial time. Therefore, the total running time of the algorithm is polynomial. 
The following theorem summarizes the result of this section.

\begin{theorem}
Let $P$ be a set of $n$ red and blue points in the plane. One can compute in polynomial time a planar bichromatic spanning tree of $P$ of bottleneck at most $8\sqrt{2}$ times the bottleneck of an optimal bichromatic spanning tree of $P$.
\end{theorem}

\old{
\section{Conclusion}
In this paper, we consider the problem of computing bottleneck planar bichromatic spanning tree of a set of bi-colored points in the plane. We show that the problem is NP-hard and present the first constant-factor approximation algorithm for the problem.
}

\bibliographystyle{plainurl}

\begin{thebibliography}{10}

\bibitem{Abellanas99}
M.~Abellanas, J.~Garcia{-}Lopez, G.~Hern{\'{a}}ndez{-}Pe{\~{n}}alver, M.~Noy,
  and P.~A. Ramos.
\newblock Bipartite embeddings of trees in the plane.
\newblock {\em Discr. Appl. Math.}, 93(2-3):141--148, 1999.

\bibitem{Abu-Affash15}
A.~K. Abu-Affash, A.~Biniaz, P.~Carmi, A.~Maheshwari, and M.~Smid.
\newblock Approximating the bottleneck plane perfect matching of a point set.
\newblock {\em Comput. Geom.}, 48(9):718--731, 2015.

\bibitem{Abu-Affash14}
A.~K. Abu-Affash, P.~Carmi, M.~J. Katz, and Y.~Trabelsi.
\newblock Bottleneck non-crossing matching in the plane.
\newblock {\em Comput. Geom.}, 47(3):447--457, 2014.

\bibitem{Agarwal90}
P.~K. Agarwal.
\newblock Partitioning arrangements of lines {II}: Applications.
\newblock {\em Discr. Comput. Geom.}, 5(1):533--573, 1990.

\bibitem{Agarwal91}
P.~K. Agarwal, H.~Edelsbrunner, and O.~Schwarzkopf.
\newblock Euclidean minimum spanning trees and bichromatic closest pairs.
\newblock {\em Discr. Comput. Geom.}, 6(1):407--422, 1991.

\bibitem{Aichholzer08}
O.~Aichholzer, S.~Cabello, R.~Fabila-Monroy, D.~Flores-Pe\~{n}aloza, T.~Hackl,
  C.~Huemer, F.~Hurtado, and D.~R. Wood.
\newblock Edge-removal and non-crossing configurations in geometric graphs.
\newblock In {\em EuroCG}, pages 119--122, 2008.

\bibitem{Alon93}
N.~Alon, S.~Rajagopalan, and S.~Suri.
\newblock Long non-crossing configurations in the plane.
\newblock In {\em SoCG}, pages 257--263, 1993.

\bibitem{Aloupis10}
G.~Aloupis, J.~Cardinal, S.~Collette, E.~D. Demaine, M.~L. Demaine, M.~Dulieu,
  R.~Fabila-Monroy, V.~Hart, F.~Hurtado, S.~Langerman, M.~Saumell, C.~Seara,
  and P.~Taslakian.
\newblock Matching points with things.
\newblock In {\em LATIN, volume 6034 of LNCS}, pages 456--467, 2010.

\bibitem{Arora04}
S.~Arora and K.~L. Chang.
\newblock Approximation schemes for degree-restricted {MST} and red–blue
  separation problems.
\newblock {\em Algorithmica}, 40(3):189--210, 2004.

\bibitem{Atallah01}
M.~J. Atallah and D.~Z. Chen.
\newblock On connecting red and blue rectilinear polygonal obstacles with
  nonintersecting monotone rectilinear paths.
\newblock {\em Int. J. Comput. Geom. Appl.}, 11(4):373--400, 2001.

\bibitem{Biniaz19}
A.~Biniaz, P.~Bose, K.~Crosbie, J.{-}L.~De Carufel, D.~Eppstein, A.~Maheshwari,
  and M.~H.~M. Smid.
\newblock Maximum plane trees in multipartite geometric graphs.
\newblock {\em Algorithmica}, 81(4):1512--1534, 2019.

\bibitem{Biniaz18}
A.~Biniaz, P.~Bose, D.~Eppstein, A.~Maheshwari, P.~Morin, and M.~H.~M. Smid.
\newblock Spanning trees in multipartite geometric graphs.
\newblock {\em Algorithmica}, 80(11):3177--3191, 2018.

\bibitem{BiniazMS14}
A.~Biniaz, A.~Maheshwari, and M.~Smid.
\newblock Bottleneck bichromatic plane matching of points.
\newblock In {\em CCCG}, 2014.

\bibitem{BoissonnatCDUY00}
J.{-}D. Boissonnat, J.~Czyzowicz, O.~Devillers, J.~Urrutia, and M.~Yvinec.
\newblock Computing largest circles separating two sets of segments.
\newblock {\em Int. J. Comput. Geom. Appl.}, 10(1):41--53, 2000.

\bibitem{Borgelt09}
M.~G. Borgelt, M.~Van~Kreveld, M.~L{\"o}ffler, J.~Luo, D.~Merrick, R.~I.
  Silveira, and M.~Vahedi.
\newblock Planar bichromatic minimum spanning trees.
\newblock {\em J. Discrete Algorithms}, 7(4):469--478, 2009.

\bibitem{DemaineEHILMOW05}
E.~D. Demaine, J.~Erickson, F.~Hurtado, J.~Iacono, S.~Langerman, H.~Meijer,
  M.~H. Overmars, and S.~Whitesides.
\newblock Separating point sets in polygonal environments.
\newblock {\em Int. J. Comput. Geom. Appl.}, 15(4):403--420.

\bibitem{Dumitrescu01}
A.~Dumitrescu and R.~Kaye.
\newblock Matching colored points in the plane: some new results.
\newblock {\em Comput. Geom.}, 19(1):69--85, 2001.

\bibitem{Dumitrescu02}
A.~Dumitrescu and J.~Pach.
\newblock Partitioning colored point sets into monochromatic parts.
\newblock {\em Int. J. Comput. Geom. Appl.}, 12(05):401--412, 2002.

\bibitem{EverettRK96}
H.~Everett, J.{-}M. Robert, and M.~J. van Kreveld.
\newblock An optimal algorithm for the ({\textless}= k)-levels, with
  applications to separation and transversal problems.
\newblock {\em Int. J. Comput. Geom. Appl.}, 6(3):247--261, 1996.

\bibitem{Kaneko03}
A.~Kaneko and M.~Kano.
\newblock Discrete geometry on red and blue points in the plane -- a survey.
\newblock {\em Discr. Comput. Geom.}, 25:551--570, 2003.

\bibitem{Lichtenstein82}
D.~Lichtenstein.
\newblock Planar formulae and their uses.
\newblock {\em SIAM J. Comput.}, 11(2):329--343, 1982.

\bibitem{Mairson88}
H.~G. Mairson and J.~Stolfi.
\newblock Reporting and counting intersections between two sets of line
  segments.
\newblock In {\em Theoretical Foundations of Computer Graphics and CAD},
  volume~40 of {\em NATO ASI Series}, pages 307--325, 1988.

\end{thebibliography}

\end{document}